\documentclass[oneside, 11pt, reqno]{amsart}
\usepackage[left=1in, right=1in, top=1in, bottom=1in]{geometry}

\usepackage{amsmath,amsthm,amsfonts,amscd} 
\usepackage{graphicx,enumitem}

\numberwithin{equation}{section}
\theoremstyle{plain}                
\newtheorem{theorem}{Theorem}[section]
\newtheorem{lemma}[theorem]{Lemma}
\newtheorem{proposition}[theorem]{Proposition}

\theoremstyle{definition}           
\newtheorem{definition}[theorem]{Definition}

\theoremstyle{remark}

\newcommand{\sgn}{\textrm{sgn}}
\newcommand{\hvpz}{\hat{\varphi}^0}
\newcommand{\hvp}{\hat{\varphi}}
\newcommand{\ox}{\overline{x}}    \newcommand{\ux}{\underline{x}}
    \newcommand{\uy}{\underline{y}}

\newcommand{\hx}{\hat{x}}

\newcommand{\ga}{g_{\alpha}}
\newcommand{\ba}{\beta_{\alpha}}

\newcommand{\bC}{\mathbb{C}}
\newcommand{\tT}{\tilde{T}}
\newcommand{\tH}{\tilde{H}}

\newcommand{\tg}{\tilde{g}}
\newcommand{\abs}[1]{\left| #1 \right|}

\newtheorem{rem}{\scshape\bf Remark}[section]

\begin{document}

\markboth{Jin Hyuk Choi}{Stochastics: An International Journal of Probability and Stochastic Processes}

\title[Asymptotic analysis for Merton's problem with transaction costs]{Asymptotic analysis for Merton's problem with transaction costs in power utility case}

\author{Jin Hyuk Choi  }

\thanks{Email: jinhyuk@andrew.cmu.edu \\ Department of Mathematical Sciences, Carnegie Mellon University, Pittsburgh, USA}

\maketitle
\begin{abstract}
We revisit the optimal investment and consumption problem with proportional transaction costs. We prove that both the value function and the slopes of the lines demarcating the no-trading region are analytic functions of cube root of the transaction cost parameter. Also, we can explicitly calculate the coefficients of the fractional power series expansions of the value function and the no-trading region. 
\end{abstract}


\section{Introduction}

In mathematical finance, the optimal investment and consumption problem has been intensively studied by many researchers since the seminal work of Merton \cite{Mert:69,Mert:71}. In \cite{Mert:69,Mert:71}, the author showed that, for power or logarithmic utilities, it is optimal to keep a constant proportion (the Merton proportion) of the wealth invested in a risky asset. 

As a generalization of \cite{Mert:69,Mert:71}, Constantinides and Magill \cite{ConMag:76} considered the model in which proportional transaction costs ($\lambda$ times transaction amounts) are imposed for each transaction. They intuited that the optimal trading strategy for power and logarithmic utilities can be described in the following way: the agent should minimally trade in such a way that the agent's proportion of wealth invested in a risky asset should be located on some interval $[\underline{\pi},\overline{\pi}]$. Davis and Norman \cite{DavNor90} proved this intuition by formulating the problem as a singular stochastic control problem. Shreve and Soner \cite{ShrSon94} subsequently extended analysis of \cite{DavNor90} by removing various technical conditions assumed in \cite{DavNor90}. 

Since an explicit formula for the solution is unknown, except the case of $\lambda=0$ (no transaction costs case), the asymptotic analysis around $\lambda=0$ for the value function and no-trading region ($\underline{\pi}$, $\overline{\pi}$) has been studied. In \cite{ShrSon94}, they showed that the effect of transaction cost on the value function (on $\underline{\pi}$ and $\overline{\pi}$, resp.) is of order $\lambda^{\frac{2}{3}}$ ($\lambda^{\frac{1}{3}}$, resp.). In this model, Jane{\v{c}}ek and Shreve \cite{JanShr04} determined the coefficients of the $\lambda^{\frac{2}{3}}$ term  ($\lambda^{\frac{1}{3}}$ term, resp.) for the expansion of the value function ($\underline{\pi}$ and $\overline{\pi}$, resp.), by constructing sub- and supersolutions of the Hamilton-Jacobi-Bellman (HJB) equation for the value function.

More recently, Kallsen and Muhle-Karbe \cite{KalMuh10} employed the concept of the shadow price for the analysis of the problem above. Conceptually, the shadow price is a risky asset price process which lies between the bid and ask price of the original market, and the frictionless market with shadow price and the original market with transaction costs lead to same value function. In \cite{KalMuh10}, the authors constructed the shadow price for the case of the logarithmic utility, under the assumption that the Merton proportion $\pi$ is less than 1. The case of power utility was considered in Herczegh and Prokaj \cite{HP2011}, and the authors showed that there exists a shadow price under the assumption that the no-trading region is inside the first quadrant of $\mathbb{R}^2$. In parallel with \cite{HP2011}, Choi, S{\^\i}rbu and {\v Z}itkovi\'c \cite{CMZ12} showed that the shadow price can be constructed for power and logarithmic utilities without assumptions in \cite{KalMuh10} and \cite{HP2011}, whenever the original problem is well-posed, i.e., the value function is finite. \cite{CMZ12} also provided explicit characterization of the finiteness of the value function, which was the only remaining technical assumption in the analysis of \cite{ShrSon94}. Gerhold, Muhle-Karbe and Schachermayer \cite{GMS11} used the shadow price approach to obtain the results of asymptotic expansion up to an arbitrary order for logarithmic utility, under the assumption $\pi<1$.

In this paper, we show that the value function and the slopes ($\underline{\pi}$, $\overline{\pi}$) of the lines which determine the no-trading wedge are analytic functions of $\lambda^{\frac{1}{3}}$ at $\lambda=0$, in case of the power utility. Furthermore, the coefficients of power series expansions of these analytic functions can be obtained up to an arbitrary order by recursive calculations. This result extends the work of \cite{GMS11} to the power utility case, without the technical condition $\pi<1$ imposed in \cite{GMS11}. Also, we confirm that the assumption in \cite[Section 3]{JanShr04} is true: $\underline{\pi}$, $\overline{\pi}$ and the value function  have expansions in powers of $\lambda^{\frac{1}{3}}$. 

The paper is organized as follows. In Section \ref{model}, we introduce Merton's problem with transaction costs, and as a preliminary, provide a proposition which is a modified version of \cite[Theorem 2.8]{CMZ12}. In Section \ref{analytic}, we prove that the solution of the free-boundary problem is analytically dependent on $\lambda^{\frac{1}{3}}$. By using the result of Section \ref{analytic}, in Section \ref{series}, we prove the main result of this paper: the value function $u$ and the slopes of the lines demarcating the no-trading region can be written as power series of $\lambda^{\frac{1}{3}}$, and the coefficients can be explicitly calculated. As an example, we provide the first several coefficients of power series expansions. Since the analysis is based on the results of \cite{CMZ12}, we present some of them in Appendix.

\section{Merton's Problem with Transaction Costs}\label{model}
The model of a financial market is the same as that of \cite{DavNor90, ShrSon94, JanShr04, GMS11, CMZ12}.
The stock (risky asset) price process $\{S_t\}_{t\geq 0}$ is given by
\begin{equation}
dS_t =S_t(\mu dt + \sigma dB_t),\,  t\geq 0 \textrm{  with  } S_0>0,
\end{equation}
where $\{B_t\}_{t\geq 0}$ is the standard Brownian motion, and $\mu>0$ and $\sigma>0$ are constants. 
An agent initially has $\eta_B>0$ units of bonds with a constant value 1, and has $\eta_S>0$ shares of stocks. The agent needs to pay transaction costs proportional to the amount of stock sold, which means that bid and ask price of stock are $(1-\lambda)S_t$ and $S_t$, respectively\footnote{As discussed in \cite{GMS11}, by scaling the stock price, we can easily see that this formulation is equivalent to the model with bid and ask price of $(1-\underline{\lambda})S_t$ and $(1+\overline{\lambda})S_t$, for constants $\underline{\lambda}\in (0,1)$ and  $\overline{\lambda}>0$.}.

We describe the agent's {\bf investment/consumption strategy} with a triple $(\varphi^0,\varphi,c)$ of optional processes (with respect to the natural augmentation of the filtration generated by B). $\varphi^0$ and $\varphi$ are right-continuous and finite variation, and they represent number of shares of bond and number of shares of stock, respectively. The process $c$ is nonnegative and locally integrable, and it represents the consumption rate. 
We assume that the initial value $(\varphi^0_{0-},\varphi_{0-})$ equals $(\eta_B,\eta_S)$. Note that $(\varphi^0_{0},\varphi_{0})$ may differ from $(\varphi^0_{0-},\varphi_{0-})$ because of the initial transaction at time zero.

For $p\in (-\infty,1)\setminus\{0\}$, we consider the utility function $U:[0,\infty)\mapsto [-\infty,\infty)$ defined by
$$U(x)=\tfrac{1}{p}x^p \textrm{  for  } x>0, 
\textrm{  and  } U(0)=\left \{ \begin{array}{ll} 0, &p>0,\\ -\infty, &p<0.\end{array} \right.$$

Now we define admissible strategies and the utility maximization problem considered in this paper.

\begin{definition}{\textit{(Admissible strategy, optimal strategy, value function)}}
We call an investment/consumption strategy $(\varphi^0,\varphi,c)$ {\bf admissible} if it satisfies self-financing and solvency conditions, i.e.,
\begin{equation}
\begin{split}
\varphi_t^0 = \varphi_{0-}^{0} - \int_0^t S_u d\varphi_u^{\uparrow} + \int_0^t (1-\lambda)S_u d\varphi_u^{\downarrow} - \int_0^t c_u du, \quad &\textrm{(self-financing)}\\
\varphi_t^0 + (\varphi_t)^+ S_t - (\varphi_t)^- (1-\lambda)S_t \geq 0, \quad &\textrm{(solvency)}
\end{split}
\end{equation}
where $\varphi=\varphi_{0-}+\varphi^{\uparrow}-\varphi^{\downarrow}$ is the pathwise minimal (Hahn-Jordan) decomposition of $\varphi$ into a difference of two non-decreasing adapted, right-continuous processes.

We consider the following optimal-consumption problem, with the value
\begin{equation}
\begin{split}\label{primal}
u(\eta_S,\eta_B) = \sup_{\textrm{admissible $(\varphi^0,\varphi,c)$}} \mathbb{E}\Big[\int_0^{\infty} e^{-\delta t} U(c_t)dt\Big],
\end{split}
\end{equation}
where the constant $\delta>0$ stands for the impatience rate. We call the admissible maximizer of the problem the {\bf optimal strategy} and denote it by  $(\hat{\varphi}^0,\hat{\varphi},\hat{c})$.
\end{definition}

From the result of \cite{ShrSon94}, there exist two constants, namely $\underline{\pi}$ and $\overline{\pi}$ (depending on market parameters), such that the optimal trading strategy $(\hat{\varphi}^0,\hat{\varphi})$ is to minimally trade in such a way that 
$$\underline{\pi}\leq\frac{\hvp_t S_t}{\hvp_t^0 + \hvp_t S_t}\leq \overline{\pi}$$
holds. In other words, $\underline{\pi}$ and $\overline{\pi}$ are slopes of the lines demarcating the no-trading wedge. 

Our goal in this paper is to show that $\underline{\pi}$, $\overline{\pi}$ and the value $u(\eta_S,\eta_B)$ can be written as power series of $\lambda^{\frac{1}{3}}$ for small enough $\lambda$, and provide a way to calculate the coefficients of those power series. For that purpose, we use the results of \cite{CMZ12}, which enable us to express $\underline{\pi}$, $\overline{\pi}$ and the value $u(\eta_S,\eta_B)$ in terms of the solution of a free-boundary problem. We adopt the notation of the constant $\pi$ (the Merton proportion) and the point $N\in \mathbb{R}^2$ in \cite{CMZ12}, i.e., 
\begin{equation}
\begin{split}\label{pi N}
\pi= \frac{\mu}{\sigma^2 (1-p)},\textrm{  and  } N=(x_N,y_N)= \Big( \frac{2 \textrm{sgn} (p) p\mu }{2\sigma^2 \delta (1-p)- p\mu^2},\frac{2 \textrm{sgn} (p)(1-p)^2\sigma^2}{2\sigma^2 \delta (1-p)- p\mu^2}\Big).
\end{split}
\end{equation}
\begin{rem}
In the frictionless case ($\lambda=0$), the no-trading wedge reduces to a single line whose slope equals $\pi$ (see \cite{Mert:71}). Also, when $\lambda=0$, 
the solution $(\ux,\ox,g)$ of the free boundary problem in \cite[Theorem 2.8]{CMZ12} degenerate to $\ux=\ox=x_N$ and $g(x_N)=y_N$. 
\end{rem}

Since we consider small enough transaction costs to investigate the asymptotic behavior, 
it is natural to assume that the problem is well-posed even though there are no transaction costs. 
In other words, we assume the finiteness of the value function for the frictionless case, which corresponds to $2\sigma^2 \delta(1-p)-p\mu^2>0$ (\cite[Remark 9.23]{KarShr98}).
Also, to reduce the technical difficulty\footnote{The case $\pi=1$ corresponds to the case of end-point singularity $\ox=x_P$ in \cite[Proposition 6.9 (2)]{CMZ12}. }, we exclude the case $\pi=1$. Throughout this paper, we assume that following holds.

\bigskip
{\bf Assumption:} $2\sigma^2 \delta(1-p)-p\mu^2>0$ and $\pi\neq 1$.

\bigskip

In the work of \cite{CMZ12}, they show that the value function $u$ in \eqref{primal} and the no-trading wedge $\overline{\pi}$, $\underline{\pi}$ are characterized by the solution of the free boundary problem given in \eqref{free boundary ODE}. For the general cases, some analysis near the point $N$ for the smoothness of the constructed solution of the free boundary problem is required. But when we consider small transaction costs, the solution of the free boundary problem is bounded away from the point $N$. Thus we can simplify \cite[Theorem 2.8]{CMZ12} in the following way. A version of \cite[Theorem 2.8]{CMZ12} can be found in Appendix.

\begin{proposition}\label{previous}
For given small enough $\lambda>0$,
\begin{enumerate}
\item There exists a unique solution $(\ux, \ox, g)$ to the following free boundary problem:
\begin{equation}
\begin{split} \label{free boundary ODE}
&g'(x)=L(x,g(x)), \  0<\ux<\ox , \   g\in C^{\infty}([\ux,\ox]),\\
&g(\ux)=T(\ux), \  g'(\ox)=0, \ \int_{\ux}^{\ox} \tfrac{g'(x)}{x}dx = \ln{(\tfrac{1}{1-\lambda})},
\end{split}
\end{equation}
where $L(x,z)$ and $T(x)$ are given by 
\begin{equation}\label{L and T} \left\{ \begin{array}{ll}L(x,z)=\frac{-\sigma^2(1-p)^3 x^2 +2p(1-p)(sgn(p)+\mu x)z-2\delta p z^2}{(1-p)x(2 sgn(p)+2\mu x + \sigma^2(p^2-1)x) - (2\delta x + p(1-p)(2 sng(p)+2\mu x - \sigma^2 x))z + 2\delta p z^2}, \\
T(x) = \frac{p(1-p)(\sgn(p)+\mu x) +(1-p)\sqrt{p(p+2\sgn(p)p\mu x-(2\sigma^2(1-p)-p\mu^2)x^2)}}{2\delta p}.
\end{array}\right.
\end{equation}
\item $\underline{\pi}$ and $\overline{\pi}$, which determine the no-trading wedge, can be expressed in terms of the solution $(\ux,\ox,g)$ of the problem \eqref{free boundary ODE}:
\begin{equation}
\begin{split}
\underline{\pi}=\frac{(1-p)\ux}{p\, g(\ux)},\quad \overline{\pi}=\frac{(1-p)\ox}{(1-p)\lambda \ox + p (1-\lambda)g(\ox)}. 
\end{split}
\end{equation}
\item The value function $u$ in \eqref{primal} can be written as,
\begin{equation}\label{expu}
u(\eta_S,\eta_B)=\tfrac{1}{p} \big(\eta_B + \eta_S S_0  e^{f(\hx)}\big)^p \vert g(\hx)\vert^{1-p},
\end{equation}
where the function $f :[\ux,\ox]\mapsto [0,\ln (1-\lambda)]$ and the point $\hx\in [\ux,\ox]$ are given by
\begin{equation}
\begin{split}
f(x)&=\ln(1-\lambda)+\int_{x}^{\ox}\frac{g'(x)}{x}dx, \\
\hx &= \left\{ \begin{array}{ll} \ox, & \textrm{if  }\ \frac{S_0 \eta_S}{\eta_B + S_0 \eta_S}>\overline{\pi}, \\ \ux, &\textrm{if  }\  \tfrac{S_0 \eta_S}{\eta_B + S_0 \eta_S}<\underline{\pi}, \\ \textrm{the solution to $\frac{(1-p)x}{p \, g(x)}=\frac{S_0 \eta_S \, e^{f(x)}}{\eta_B + S_0 \eta_S \, e^{f(x)}}$}, &\textrm{otherwise}. \end{array} \right.
\end{split}
\end{equation}
\end{enumerate}
\end{proposition}

\begin{proof}
(1) The existence of the solution of the free boundary problem in part (1) of Theorem \ref{thm:main} is proved in \cite[Section 6]{CMZ12}, with the formulation $g'(x)=L(x,g(x))$ and $g(\ux)=T(\ux)$. In fact, after substituting the minimizers of the infimum in \eqref{equ:HJB-g}, we obtain \eqref{free boundary ODE}. And the uniqueness of the solution $(\ux,\ox,g)$ can be deduced by \cite[Remark 6.15]{CMZ12}. 
Thus, we only need to show that $g\in C^{\infty}([\ux,\ox])$. If $\pi<1$, then $g\in C^{\infty}([\ux,\ox])$ follows from \cite[Proposition 6.7]{CMZ12}. 
In case $\pi>1$, $g\in C^{\infty}([\ux,\ox])$ follows from \cite[Proposition 6.7 (1) a)]{CMZ12}, because small enough $\lambda$ implies that $\ux\in (x_P,x_N)$ there.

\medskip

\noindent(2) From the expression of the optimal proportion process in part (3) of Theorem \ref{thm:main}, we observe that $\underline{\pi}$ and $\overline{\pi}$ are minimum and maximum values of $\frac{(1-p)x}{(1-p)(1-e^{f(x)})x + p\, g(x) e^{f(x)}}$ for $x\in [\ux,\ox]$, which are
\begin{equation}\nonumber
\begin{split}
\underline{\pi}&=\frac{(1-p)\ux}{(1-p)(1-e^{f(\ux)})\ux + p\, g(\ux) e^{f(\ux)}}=\frac{(1-p)\ux}{p\, g(\ux)},\\
\overline{\pi}&=\frac{(1-p)\ox}{(1-p)(1-e^{f(\ox)})\ox + p\, g(\ox) e^{f(\ox)}}=\frac{(1-p)\ox}{(1-p)\lambda \ox + p (1-\lambda)g(\ox)}. 
\end{split}
\end{equation}
Indeed, by the direct calculation using $g(x)=L(x,g(x))$, we observe that $p\, g(x)(1+g'(x))-x \, g'(x)>0$ and  
$$\frac{d}{dx}\Big(\frac{(1-p)x}{(1-p)(1-e^{f(x)})x + p\, g(x) e^{f(x)}}\Big)=\frac{(1-p)e^{f(x)}(p\, g(x)(1+g'(x))-x \, g'(x))}{((1-p)(1-e^{f(x)})x + p\, g(x) e^{f(x)})^2} >0.$$ 
Thus, $\frac{(1-p)x}{(1-p)(1-e^{f(x)})x + p\, g(x) e^{f(x)}}$ takes its maximum and minimum at $\ox$ and $\ux$, respectively.

\medskip

\noindent(3) The expression \eqref{expu} of the value function $u$ is given in part (3) of Theorem \ref{thm:main}.
\end{proof}

\begin{rem}\label{g alpha}
The proof of part (1) of Theorem \ref{thm:main} consists of two steps:
\begin{enumerate}
\item Fix $\alpha\in(0,x_N)$. Solve the ODE $\ga'(x)=L(x,\ga(x))$ with initial condition $\ga(\alpha)=T(\alpha)$ and let it evolve to the right until $\ba=\inf\big\{x>\alpha: \ga'(x)=0\big\}$.
\item Vary $\alpha$ to meet the integral condition $\int_{\alpha}^{\ba} \frac{\ga'(x)}{x}dx = \log{(\frac{1}{1-\lambda})}$.
\end{enumerate}
\end{rem}

\section{The Solution's Analytic Dependence on $\lambda^{\frac{1}{3}}$ }\label{analytic}
Proposition \ref{previous} says that information about the value function $u$ and the no-trading-wedge is encoded in the solution $(\ux,\ox,g)$ of the free boundary problem \eqref{free boundary ODE}. Thus, we concentrate on the asymptotic behavior of $(\ux,\ox,g)$ in terms of $\lambda$. In Lemma \ref{holomorphic}, motivated by the previous remark, we first consider an initial value problem 
\begin{equation}\label{IVP}
\ga'(x)=L(x,\ga'(x)), \quad \ga(\alpha)=T(\alpha),
\end{equation}
and show that $\ga$ and $\ba$ are analytic functions of $\alpha$. Then in Proposition \ref{lambda}, we show that $\alpha$ is an analytic function of $\lambda^{\frac{1}{3}}$ and conclude that the solution $(\ux,\ox,g)$ analytically depends on $\lambda^{\frac{1}{3}}$.

For convenience, we change notation as $\ga(x)=g(x,\alpha), \ba=\beta(\alpha)$. Then \eqref{IVP} amounts to
\begin{equation}
\begin{split} \label{ODE2}
\tfrac{\partial g}{\partial x}(x,\alpha)=L(x,g(x,\alpha)), \quad g(\alpha,\alpha)=T(\alpha), \quad \tfrac{\partial g}{\partial x}(\beta(\alpha),\alpha)=0.
\end{split}
\end{equation}
We study above ODE in the complex variable framework, and we use the following notation for a complex-polydisk: for $(z,w)\in \bC^2$, $\epsilon>0$ and $\epsilon'>0$, we denote $D_{\epsilon,\epsilon'}(z,w) \in \bC^2$ and $I_{\epsilon'}(w)\in \bC$ as
\begin{displaymath}
\begin{split}
 D_{\epsilon,\epsilon'}(z,w)&=\{(z_1,z_2)\in\bC^2 : \vert z_1-z \vert<\epsilon, \ \vert z_2-w \vert<\epsilon' \}, \\
I_{\epsilon'}(w) &=\{z_2\in \bC : \vert z_2-w \vert <\epsilon'\}.
\end{split}
\end{displaymath}
Before we start Lemma \ref{holomorphic}, let $\tT(z)$ be a holomorphic extension of the real function $T$ in \eqref{L and T} on the complex-neighborhood of $z=x_N$, which is possible since the inside of the square root is non-zero for $x=x_N$.

\begin{lemma}\label{holomorphic}
There exist constants $\epsilon\geq \epsilon'>0$ and holomorphic functions $\tg:D_{\epsilon,\epsilon'}(x_N,x_N)\mapsto \bC$ and $\tilde{\beta}:I_{\epsilon'}(x_N)\mapsto I_{\epsilon}(x_N)$ such that 
\begin{enumerate}
\item $\tfrac{\partial \tg}{\partial z_1}(z_1,z_2)=L(z_1,\tg(z_1,z_2)) \ $ and $ \ \  \tg(z_2,z_2)=\tT(z_2)$.
\item $\tilde{\beta}(z_2)\neq z_2$ for $z_2 \neq z_N$, and $\ \tfrac{\partial \tg}{\partial z_1}(\tilde{\beta}(z_2),z_2)=0$.
\end{enumerate}
\end{lemma}

\begin{proof}
(1) In \eqref{L and T} we observe that $L(z_1,z_2)$ is holomorphic near $N$, as a complex variable function. By \cite[Theorm 1.1 and Remark 1.6]{ODE1}, the initial value problem
\begin{equation}
\begin{split}\label{h}
h'(z_1)=L(z_1,h(z_1)), \quad h(z_2)=z_3
\end{split}
\end{equation}
has a unique solution $h(z_1,z_2,z_3)$ which is defined and holomorphic on some complex-neighborhood of $(x_N,x_N,y_N)\in \bC^3$. We define $\tg$ as
\begin{equation}
\begin{split}
\tg(z_1,z_2)=h(z_1,z_2,\tT(z_2)).
\end{split}
\end{equation}
Since the map $(z_1,z_2)\mapsto (z_1,z_2,\tT(z_2))$ is holomorphic at $(x_N,x_N)\in \bC^2$, $\tg$ is holomorphic at $(x_N,x_N)\in \bC^2$. Thus, there exists $\epsilon>0$ such that $\tg(z_1,z_2)$ is holomorphic on $D_{\epsilon,\epsilon}(x_N,x_N)$.
Obviously, the definition of $\tg$ and \eqref{h} implies that, on $D_{\epsilon,\epsilon}(x_N,x_N)$,
\begin{equation}\label{tg}
\tfrac{\partial \tg}{\partial z_1}(z_1,z_2)=L(z_1,\tg(z_1,z_2)), \textrm{  and  } \tg(z_2,z_2)=\tT(z_2).
\end{equation}

\medskip

\noindent(2) Our goal is to define $\tilde{\beta}(z_2)$ as a solution of $\tfrac{\partial \tg}{\partial z_1}(\cdot,z_2)=0$, different from $z_2$. For that purpose, we first consider a holomorphic function $H(z_1,z_2)$ defined on $D_{\epsilon,\epsilon}(x_N,x_N)\setminus \{(z_1,z_2)\in \bC^2 : z_1=z_2\}$ as
\begin{equation}
\begin{split}\label{def H}
H(z_1,z_2) = \frac{\frac{\partial \tg}{\partial z_1}(z_1,z_2)}{z_1-z_2}.
\end{split}
\end{equation}
We observe that $H$ is locally bounded at $\{(z_1,z_2)\in \bC^2 : z_1=z_2\}$, because
\begin{equation}
\begin{split}\label{H}
\lim_{z_1 \to z_2} H(z_1,z_2) = \lim_{z_1 \to z_2} \frac{\frac{\partial \tg}{\partial z_1}(z_1,z_2)-\frac{\partial \tg}{\partial z_1}(z_2,z_2)}{z_1-z_2}=\tfrac{\partial^2 \tg}{\partial z_1^2} (z_2,z_2),
\end{split}
\end{equation}
where the first equality followed by an obervation
\begin{equation}\label{g zero}
\tfrac{\partial \tg}{\partial z_1}(z_2,z_2)=L(z_2,\tg(z_2,z_2))=L(z_2,\tT(z_2))=0.
\end{equation}
Since codimension of the set $\{(z_1,z_2)\in \bC^2 : z_1=z_2\}$ in $\bC^2$ is 1, we can use the Riemann Removable Singularity Theorem(\cite[Theorem 4.2.1]{ODE2}) to conclude that $H$ has the unique holomorphic continuation $\tH:D_{\epsilon,\epsilon}(x_N,x_N)\mapsto \bC$. To apply the Implicit Function Theorem, we show that $\tH(x_N,x_N)=0$ and $\tfrac{\partial \tH}{\partial z_1}(x_N,x_N)\neq 0$.\\
i) $\tH(x_N,x_N)=0$: By using \eqref{H} and \eqref{tg}, $\tg(x_N,x_N)=y_N$ and $\tfrac{\partial \tg}{\partial z_1}(x_N,x_N)=0$,
\begin{equation}
\begin{split}
\tH(x_N,x_N)&=\lim_{z_1 \to x_N} \tH(z_1,x_N) = \lim_{z_1 \to x_N} H(z_1,x_N) \\
&= \tfrac{\partial^2 \tg}{\partial z_1^2} (x_N,x_N) = \tfrac{d }{d z_1} \Big(L(z_1,\tg(z_1,x_N))\Big)\Big\vert_{z_1=x_N} = 0.
\end{split}
\end{equation}
ii) $\tfrac{\partial \tH} {\partial z_1}(x_N,x_N)\neq 0$: By using $\tfrac{\partial \tg}{\partial z_1}(x_N,x_N)=\tfrac{\partial^2 \tg}{\partial z_1^2} (x_N,x_N)=0$ and $\tg(x_N,x_N)=y_N$,
\begin{equation}
\begin{split}
&\tfrac{\partial \tH} {\partial z_1}(x_N,x_N) =\lim_{z_1\to x_N} \tfrac{\partial H} {\partial z_1}(z_1,x_N)=\lim_{z_1\to x_N} \tfrac{(z_1-x_N)\frac{\partial^2 \tg}{\partial z_1^2}(z_1,x_N) - \frac{\partial \tg}{\partial z_1}(z_1,x_N) }{(z_1-x_N)^2} \\
&\qquad=\lim_{z_1\to x_N} \tfrac{\frac{\partial^2 \tg}{\partial z_1^2}(z_1,x_N)-\frac{\partial^2 \tg}{\partial z_1^2}(x_N,x_N)   }{z_1-x_N}  -  \tfrac{\frac{\partial \tg}{\partial z_1}(z_1,x_N)-\frac{\partial \tg}{\partial z_1}(x_N,x_N)- \frac{\partial^2 \tg}{\partial z_1^2}(x_N,x_N)(z_1-x_N)}{(z_1-x_N)^2} \\
&\qquad=\tfrac{1}{2} \tfrac{\partial^3 \tg}{\partial z_1^3}(x_N,x_N)  = \tfrac{1}{2} \tfrac{d^2 }{d z_1^2} \Big( L(z_1,\tg(z_1,x_N))\Big) \Big\vert_{z_1=x_N} \\
&\qquad=-\tfrac{(2\sigma^2\delta(1-p)-p \mu^2)^2}{2p^2 \sigma^2 \mu (\pi-1)^2} \neq 0.
\end{split}
\end{equation}

Now we can apply the holomorphic version of the Implicit Function Theorem (\cite[Theorem 3.1.4]{ODE2}) to obtain a constant $\epsilon' \in (0,\epsilon)$ and a holomorphic function $\tilde{\beta}: I_{\epsilon'}(x_N) \mapsto I_{\epsilon}(x_N)$ such that
\begin{equation}
\begin{split}
\tH(\tilde{\beta}(z_2),z_2)=0 \ \textrm{  and  }\ \tilde{\beta}(x_N)=x_N.
\end{split}
\end{equation}
From a direct calculation using \eqref{tg}, \eqref{H} and \eqref{g zero}, we observe that
$$\tH(z_2,z_2)=0 \iff \tT(z_2)=\tfrac{(1-p)^2\sigma^2 z_2}{p\mu} \iff z_2 = 0 \textrm{  or  } x_N.$$
Hence, $\tilde{\beta}(z_2)\neq z_2$ for $z_2 \neq x_N$.

Finally, it remains to show that $\tfrac{\partial\tg}{\partial z_1}(\tilde{\beta}(z_2),z_2)=0$ for $z_2 \in I_{\epsilon'}(x_N)$. In case $z_2=x_N$, $\tfrac{\partial \tg}{\partial z_1}(\tilde{\beta}(x_N),x_N)=\tfrac{\partial \tg}{\partial z_1}(x_N,x_N) =0$. 
If $z_2 \neq x_N$, then $\tilde{\beta}(z_2)\neq z_2$, so, 
$$0=\tH(\tilde{\beta}(z_2),z_2)=H(\tilde{\beta}(z_2),z_2)=\frac{\frac{\partial \tg}{\partial z_1}(\tilde{\beta}(z_2),z_2)}{\tilde{\beta}(z_2)-z_2},$$
which implies that $\tfrac{\partial \tg}{\partial z_1}(\tilde{\beta}(z_2),z_2)=0$.
\end{proof}

\begin{rem}[(Coefficients of the power series expansions of $\tg$ and $\tilde{\beta}$)]\label{g z1z2}
In Lemma \ref{holomorphic}, $\tg$ and $\tilde{\beta}$ are holomorphic on $D_{\epsilon,\epsilon'}(x_N,x_N)$ and $\ I_{\epsilon'}(x_N)$, respectively. So, they can be written as power series on their holomorphic domains;
\begin{equation}
\begin{split}\label{g aij}
\tg(z_1,z_2)=\sum_{i\geq 0} \sum_{j\geq 0}a_{ij}(z_1-x_N)^i (z_2-x_N)^j, \quad \tilde{\beta}(z_2)=\sum_{i\geq0} b_i (z_2-x_N)^i.
\end{split}
\end{equation}
To detemine the coefficients $a_{ij}$'s, we substitute the expression
\begin{equation}
\begin{split}
\tg(z_1,z_2)&=\sum_{i=0}^n \sum_{j=0}^{n-1}a_{ij}(z_1-x_N)^i (z_2-x_N)^j + O((z_1-x_N)^{n+1}+(z_2-x_N)^n),\\
\tfrac{\partial\tg}{\partial z_1}(z_1,z_2) &= \sum_{i=0}^{n} \sum_{j=0}^{n-1}ia_{ij}(z_1-x_N)^{i-1} (z_2-x_N)^j + O((z_1-x_N)^n+(z_2-x_N)^n),\\
\end{split}
\end{equation}
into \eqref{tg}. Then, the first one generates $n^2$ equations and the second one generates $n$ equations of $\{a_{ij}\}_{0\leq i \leq n, 0\leq j \leq n-1}$. Since we have $n(n+1)$ equations of $n(n+1)$ unknowns, we can calculate $\{a_{ij}\}_{0\leq i \leq n, 0\leq j \leq n-1}$ for any $n$. Similarly, we can calculate $b_i$'s from $\tfrac{\partial \tg}{\partial z_1}(\tilde{\beta}(z_2),z_2)=0$. 
\end{rem}

\begin{proposition}\label{lambda}
Let $\tg:D_{\epsilon}(x_N,x_N) \mapsto \bC$, $\tilde{\beta}:I_{\epsilon'}(x_N)\mapsto I_{\epsilon}(x_N)$ be as in Lemma \ref{holomorphic}. Then, there exist a constant $\epsilon''>0$ and a holomorphic function $\alpha:I_{\epsilon''}(0)\mapsto I_{\epsilon'}(x_N)$ which satisfies following property; for any positive real number $\lambda<\epsilon''$, the solution $(\ux,\ox,g)$ of the free boundary problem \eqref{free boundary ODE} can be written as
\begin{equation}
\begin{split}
\ux = \alpha(\lambda^{\frac{1}{3}}),\ \ox= \tilde{\beta}(\alpha(\lambda^{\frac{1}{3}})), \textrm{  and  }g(x) = \tg(x,\alpha(\lambda^{\frac{1}{3}})) \textrm{  for  } x\in [\ux,\ox].
\end{split}
\end{equation}
\end{proposition}
\begin{proof}
Since $\frac{1}{z_1} \cdot \frac{\partial\tg}{\partial z_1}(z_1,z_2) $ is holomorphic on $D_{\epsilon,\epsilon'}(x_N,x_N)$ (if necessary, we may reduce $\epsilon$), there exists a holomorphic function $G:D_{\epsilon,\epsilon'}(x_N,x_N)\mapsto \bC$ such that
\begin{equation}\label{defG}
\tfrac{\partial G}{\partial z_1}(z_1,z_2)=\tfrac{1}{z_1} \cdot \tfrac{\partial\tg}{\partial z_1}(z_1,z_2).
\end{equation}
By the method used in Remark \ref{g z1z2}, we can find the following series expansion valids on $I_{\epsilon'}(x_N)$,
\begin{equation}\label{expG}
1-\exp \Big( G(z_2,z_2)-G(\tilde{\beta}(z_2),z_2)\Big) = (z_2-x_N)^3 \big( c_0 + \sum_{i\geq 1} c_i (z_2-x_N)^i \big),
\end{equation}
where $c_0 = -\tfrac{(2\sigma^2 \delta (1-p) - p \mu^2 )^3}{6p^3 \sigma^2 (\pi-1)^2 \mu^2} \neq 0$. Therefore, we can find a holomorphic function $F$ on some complex-neighborhood of $x_N$ such that
\begin{equation}\label{expF}
F(z_2)^3 = c_0 + \sum_{i\geq 1} c_i (z_2-x_N)^i ,  \textrm{  and  } F(x_N)=c_0^{\frac{1}{3}}.
\end{equation}
We observe that $\frac{d}{d z_2} \big((z_2-x_N) F(z_2)\big) \Big\vert_{z_2=x_N}=c_0^{\frac{1}{3}} \neq 0$. Hence, by the Inverse Function Theorem (\cite[Theorem 3.1.4]{ODE2}),
there exist $\epsilon''>0$ and a holomorphic function $\alpha:I_{\epsilon''}(0)\mapsto I_{\epsilon'}(x_N)$ such that
\begin{equation}\label{inverseF}
(\alpha(\omega)-x_N)F(\alpha(\omega))=\omega, \textrm{  and  } \alpha(0)=x_N.
\end{equation}

In Remark \ref{g z1z2}, we can easily observe that $a_{ij}$'s and $b_i$'s are real numbers, since the coefficients of $L$ and $\tT$ are real numbers. Thus, the holomorphic functions $\tg,\tilde{\beta}, G,F$ and $\alpha$ restricted to real variables are real analytic functions. Then, for a positive real number $\lambda<\epsilon''$,
\begin{equation}
\begin{split}
\lambda &= (\alpha(\lambda^{\frac{1}{3}})-x_N)^3 F(\alpha(\lambda^{\frac{1}{3}}))^3 \\
&= 1-\exp \Big( G(z_2,z_2)-G(\tilde{\beta}(z_2),z_2)\Big)\Big\vert_{z_2=\alpha(\lambda^{\frac{1}{3}})}\\
&=1-\exp \Big( -\int_{ \alpha(\lambda^{\frac{1}{3}})}^{\tilde{\beta}(\alpha(\lambda^{\frac{1}{3}}))} \tfrac{1}{x}\cdot \tfrac{\partial \tg}{\partial z_1}(x,\alpha(\lambda^{\frac{1}{3}})) dx \Big),
\end{split}
\end{equation}
which is equivalent to 
$$\int_{ \alpha(\lambda^{\frac{1}{3}})}^{\tilde{\beta}(\alpha(\lambda^{\frac{1}{3}}))}  \tfrac{1}{x}\cdot \tfrac{\partial \tg}{\partial z_1}(x,\alpha(\lambda^{\frac{1}{3}})) dx=\log \big(\tfrac{1}{1-\lambda}\big).$$
Furthermore, from Lemma \ref{holomorphic}, $\tg$ and $\tilde{\beta}$ satisfy
\begin{multline}\nonumber
\tfrac{\partial \tg}{\partial z_1}(x,\alpha(\lambda^{\frac{1}{3}}))=L(x,\tg(x,\alpha(\lambda^{\frac{1}{3}}))), \quad \tg(\alpha(\lambda^{\frac{1}{3}}),\alpha(\lambda^{\frac{1}{3}}))=T(\alpha(\lambda^{\frac{1}{3}})), \\
\textrm{and  } \ \tfrac{\partial \tg}{\partial z_1}(\tilde{\beta}(\alpha(\lambda^{\frac{1}{3}})),\alpha(\lambda^{\frac{1}{3}}))=0.
\end{multline}
Thus, from the uniqueness of the solution $(\ux,\ox,g)$ of \eqref{free boundary ODE}, the proof is done, if we show that $\alpha(\lambda^{\frac{1}{3}})<\tilde{\beta}(\alpha(\lambda^{\frac{1}{3}}))$. For small enough positive $\lambda$, $\alpha(\lambda^{\frac{1}{3}})<x_N$ because $\alpha'(0)=c_0^{-\frac{1}{3}}<0$.
Since $\tilde{\beta}'(x_N)=-1$ (from a direct calculation, $b_1 =-1$ in \eqref{g aij}), $\tilde{\beta}(\alpha(\lambda^{\frac{1}{3}}))>x_N$ for small enough $\lambda>0$, and we conclude that $\alpha(\lambda^{\frac{1}{3}})<\tilde{\beta}(\alpha(\lambda^{\frac{1}{3}}))$. If necessary, we may reduce $\epsilon''>0$ to satisfy these arguments.
\end{proof}

\begin{rem}\label{alpha}
In \eqref{expG}, we can compute $c_i$'s by computing $a_{ij}$'s and $b_i$'s and applying them to \eqref{defG}. We can also calculate the series expansion of $F$ from \eqref{expF}. Finally, we can compute the coefficients $\{d_i\}_{i\geq 1}$ of the expression $\alpha(z)=x_N + \sum_{i\geq 1} d_i z^i$, from the fact that $\alpha$ is an inverse function of the map $z_2 \mapsto (z_2-x_N) F(z_2)$.
\end{rem}

\section{Main Result: Asymptotic Analysis}\label{series}

Now we present result about asymptotic analysis of the no-trading wedge and the value function.

\begin{theorem}{(Expansion of the no-trade region)}
\label{no-trading expansion}
 $\underline{\pi}$ and $\overline{\pi}$, which determine the no-trading region, can be written as power series of $\lambda^{\frac{1}{3}}$ which holds for small enough $\lambda>0$. To be specific, 
\begin{equation}
\begin{split}\label{asymp pi}
\underline{\pi} = \sum_{i\geq0} \underline{s}_i (\lambda^{\frac{1}{3}})^i 
= \pi  - \big(\tfrac{3\pi^2(1-\pi)^2}{4(1-p)}  \big)^{\frac{1}{3}} \lambda^{\frac{1}{3}} 
- \tfrac{\pi (2\sigma^2 \delta (1-p)- p\mu^2)}{2  \sigma^4 (1-p)^2 \big( 6\pi^2 (1-\pi)^2 (1-p)^2 \big)^{\frac{1}{3}}} \lambda^{\frac{2}{3}} 
+ O(\lambda),\\
\overline{\pi} = \sum_{i\geq0} \overline{s}_i (\lambda^{\frac{1}{3}})^i 
= \pi  + \big(\tfrac{3\pi^2(1-\pi)^2}{4(1-p)}  \big)^{\frac{1}{3}} \lambda^{\frac{1}{3}} 
- \tfrac{\pi (2\sigma^2 \delta (1-p)- p\mu^2)}{2  \sigma^4 (1-p)^2 \big( 6\pi^2 (1-\pi)^2 (1-p)^2 \big)^{\frac{1}{3}}} \lambda^{\frac{2}{3}} 
+ O(\lambda). 
\end{split}
\end{equation}
If one wants to compute more $\underline{s}_i,\overline{s}_i$'s, it can be done by recursive calculations. 
\end{theorem}

\begin{proof}
By applying of Proposition \ref{lambda} to the expression Proposition \ref{previous} (ii), we can rewrite $\underline{\pi}$ and $\overline{\pi}$ as
\begin{equation}\nonumber
\begin{split}
\underline{\pi}=\frac{(1-p)\alpha(\lambda^{\frac{1}{3}})}{p \ \tg(\alpha(\lambda^{\frac{1}{3}}), \alpha(\lambda^{\frac{1}{3}}))}, \quad
\overline{\pi}= \frac{(1-p) \tilde{\beta}( \alpha(\lambda^{\frac{1}{3}}))}{(1-p)\lambda \tilde{\beta}( \alpha(\lambda^{\frac{1}{3}})) + p(1-\lambda) \tg(\tilde{\beta}( \alpha(\lambda^{\frac{1}{3}})),\alpha(\lambda^{\frac{1}{3}}))}.
\end{split}
\end{equation}
Therefore, $\underline{\pi}$ and $\overline{\pi}$ are analytic functions of $\lambda^{\frac{1}{3}}$, and they admit the power series expansions \eqref{asymp pi}. Furthermore, we can calculate their coefficients by applying \eqref{g aij} and Remark \ref{alpha}.
\end{proof}

\medskip

To deal with the asymptotic analysis of the value function $u(\eta_B,\eta_S)$, we use the expression \eqref{expu}. In \eqref{expu}, once we find an expression of $\hat{x}$ in terms of $\lambda$, then $f(\hx)$ and $g(\hx)$ can be written by Proposition  \ref{lambda}. Note that the definition of $\hx$  in Proposition \ref{previous} clearly depends on the initial position $(\eta_B,\eta_S)$. 

\begin{theorem}\label{value expansion}{(Expansion of the value function)} Let $\lambda>0$ be small enough. Then, the value function $u(\eta_B,\eta_S)$ can be written as a power series of $\lambda^{\frac{1}{3}}$:
\begin{displaymath}
\begin{split}
u(\eta_S,\eta_B) =\sum_{i\geq0} \zeta_i (\lambda^{\frac{1}{3}})^i,
\end{split}
\end{displaymath}
where we can compute $\zeta_i$'s by reculsive calculations. For example, $\zeta_0, \zeta_1, \zeta_2$ and $\zeta_3$ are given by
\begin{displaymath}
\begin{split}
\zeta_0 &= \tfrac{\abs{y_N}^{1-p} (\eta_B + S_0 \eta_S)^p}{p}, \\
\zeta_1 &=0, \\
\zeta_2 &= -  \tfrac{\big(\tfrac{9}{128}\big)^{\frac{1}{3}} \sigma^2 \big( (1-p)\pi(\pi-1)\big)^{\frac{4}{3}}\abs{y_N}^{2-p} (\eta_B + S_0 \eta_S)^p}{1-p} , \\
\zeta_3 &=\left \{ \begin{array}{ll}
- \tfrac{\pi \abs{y_N}^{1-p}  (\eta_B + S_0 \eta_S)^p}{2},  &\textrm{if  } \frac{S_0 \eta_S}{\eta_B + S_0 \eta_S} =\pi, \\
- \tfrac{\pi \abs{y_N}^{1-p}  (\eta_B + S_0 \eta_S)^p}{2},  &\textrm{if  }  \tfrac{S_0 \eta_S}{\eta_B + S_0 \eta_S}<\pi, \\
 \tfrac{\big(\pi- \frac{S_0 \eta_S}{2(\eta_B + S_0 \eta_S)}\big)\abs{y_N}^{1-p}  (\eta_B + S_0 \eta_S)^p}{2},  &\textrm{if  } \frac{S_0 \eta_S}{\eta_B + S_0 \eta_S}>\pi. \\
\end{array} \right.
\end{split}
\end{displaymath}
\end{theorem}

\begin{proof}
In Theorem \ref{no-trading expansion}, we observe that as $\lambda \searrow 0$, $\underline{\pi}\nearrow \pi$ and $\overline{\pi}\searrow \pi$. 
Thus, from the definition of $\hx$ in Proposition \ref{previous}, for small enough $\lambda>0$,
\begin{enumerate}
\item if $\frac{S_0 \eta_S}{\eta_B + S_0 \eta_S}>\pi$, then $\hx=\ox$.
\item if $\frac{S_0 \eta_S}{\eta_B + S_0 \eta_S}<\pi$, then $\hx=\ux$.
\item if $\frac{S_0 \eta_S}{\eta_B + S_0 \eta_S}=\pi$, then $\hx$ is a solution of $\frac{(1-p)x}{p \, g(x)}=\frac{S_0 \eta_S \, e^{f(x)}}{\eta_B + S_0 \eta_S \, e^{f(x)}}$.
\end{enumerate}
Since (iii) is the only nontrivial case, now we consider the case $\frac{S_0 \eta_S}{\eta_B + S_0 \eta_S}=\pi$. We consider a holomorphic function $\tilde{r}$ on a complex-neighborhood of $(x_N,x_N)$ such that
$$\tilde{r}(z_1, z_2) = \Big(1-\frac{(1-p)z_1}{p \, \tg(z_1,z_2)}\Big) \Big(\eta_B + \eta_S S_0 \, e^{-G(z_1,z_2)+G(z_2,z_2)}\Big )- \eta_B,$$
where $G$ is as in \eqref{expG}. Then, $\tilde{r}(x_N,x_N)=0$ and $\tfrac{\partial \tilde{r}}{\partial z_1}(x_N,x_N)\neq 0$ follow from a direct calculation using $\tg(x_N,x_N)=y_N$ and $\tfrac{\partial\tg}{\partial z_1}(x_N,x_N)=0$. By the holomorphic version of the Implicit Function Theorem (\cite[Theorem 3.1.4]{ODE2}), there exists a holomorphic function $\tilde{x}$ defined on a complex-neighborhood of $x_N$ such that $\tilde{r}(\tilde{x}(z_2),z_2)=0$. Thus, by Proposition \ref{lambda} and the definition of $\hx$, we conclude that $\hx=\tilde{x}(\alpha(\lambda^{\frac{1}{3}}))$. 

Hence, $\hx$ is an analytic function of $\lambda^{\frac{1}{3}}$ for small enough $\lambda$, for all cases above ($\ux$ and $\ox$ are also analytic functions of $\lambda^{\frac{1}{3}}$, by Proposition \ref{lambda}). Now we can rewrite the expression \eqref{expu} as
\begin{displaymath}
\begin{split}
u(\eta_B,\eta_S)=
\left \{ \begin{array}{ll} \tfrac{1}{p}(\eta_B + S_0 \eta_S \, e^{G(\ux,\ux)-G(\hx,\ux)}  )^p \abs{\tg(\hx,\ux)}^{1-p}, & \frac{S_0 \eta_S}{\eta_B + S_0 \eta_S} =\pi, \\
\tfrac{1}{p}(\eta_B + S_0 \eta_S)^p \abs{\tg(\ux,\ux)}^{1-p}, & \tfrac{S_0 \eta_S}{\eta_B + S_0 \eta_S}<\pi, \\
\tfrac{1}{p}(\eta_B + (1-\lambda)S_0 \eta_S)^p \abs{\tg(\ox,\ux)}^{1-p}, & \frac{S_0 \eta_S}{\eta_B + S_0 \eta_S}>\pi.
\end{array} \right.
\end{split}
\end{displaymath}
Since $G$ and $\tg$ are holomorphic and $\ux ,\ox$ and $\hx$ are analytic functions of $\lambda^{\frac{1}{3}}$, we conclude that $u(\eta_B,\eta_S)$ is an analytic function of $\lambda^{\frac{1}{3}}$ and can be written as 
\begin{displaymath}
\begin{split}
u(\eta_S,\eta_B) =\sum_{i\geq0} \zeta_i (\lambda^{\frac{1}{3}})^i,
\end{split}
\end{displaymath}
for small enough $\lambda$. We can calculate the coefficients of the above power series as before.
\end{proof}

\begin{rem}
By the exactly same method used here, we can also get results of asymptotic analysis equivalent to Theorem \ref{no-trading expansion} and \ref{value expansion} for the case of the logarithmic utility under the assumption that $\pi\neq 1$. This extends work of \cite{GMS11}, which considered the case of the logarithmic utility under the assumption that $\pi<1$.
\end{rem}

\bibliographystyle{plain}  
\bibliography{Asymptotic_references}        

\appendix
\section{}
We summarize some results in \cite{CMZ12} which are used in this paper. In \cite{CMZ12}, they derive the HJB equation from the observation that the shadow process should be the minimizer of certain stochastic control problem. And the HJB equation can be transformed to the free boundary problem of the form \eqref{equ:HJB-g}. After showing that there exists a solution of \eqref{equ:HJB-g}, they construct the shadow price process and find an expression for the value function and the optimal strategy. We present this result in the following theorem, which is weaker than the original one \cite[Theorem 2.8]{CMZ12}, but simpler to present.

\begin{theorem}{(For $p\neq 0$ with $\underline{\lambda}=\lambda$ and $\overline{\lambda}=0$)}\\
\label{thm:main}
Assume that $\mu,\sigma\in
  (0,\infty)$ satisfy $2\sigma^2 \delta(1-p)-p\mu^2>0$. Then\\
(1) There exist constants $\ux,\ox$ with $0< \ux < \ox$ and a function $g\in C^2[\ux,\ox]$ such that
   \begin{enumerate}
  \item[a)] $g'(x)>0$ for $x\in (\ux,\ox)$, and $g$ satisfies the equation 
\begin{equation}
   \label{equ:HJB-g}
   \begin{split}
     \inf_{\Sigma,\theta\in \mathbb{R}} \Big( \tfrac{1}{2} \Sigma^2 \tfrac{x}{g'(x)} -
     \alpha_q(\Sigma,\theta) x - \beta(\theta) g(x) + \gamma(\theta)
	 \Big)=0,\ x\in (\ux,\ox),
   \end{split}
\end{equation}
where 
\begin{equation}%
    \label{equ:greeks}
    \begin{split}
& q=\tfrac{p}{1-p}, \ 
      \alpha_q(\Sigma,\theta) = \theta\sigma -\mu - \Sigma
      \Big(\tfrac{1}{2} \Sigma+\sigma-\theta(1+q) \Big),\\
      &\beta(\theta) =(1+q) \Big(\delta - \tfrac{1}{2} q  \theta^2\Big),\text{ and } 
      \gamma(\theta) =\sgn(p).
\end{split}
\end{equation}
  \item[b)] \label{ite:main-4} the following boundary/integral conditions
	are satisfied:
	\begin{equation}\label{equ:integral-cond}
	 \begin{split}
	   g'(\ux+)=g'(\ox-)=0\text{ and }
	\int_{\ux}^{\ox} \tfrac{g'(x)}{x}\,
	dx=\log(\tfrac{1}{1-\lambda}).
	 \end{split}
	\end{equation}
  \item[c)] The function $h:[\ux,\ox]\to\mathbb{R}$, defined by 
\begin{equation}
     h(x)= q g(x) \left(g'(x)+1\right)-(q+1) x g'(x)
\end{equation}
admits no zeros on $[\ux,\ox]$. 
\end{enumerate}
(2) For any $(\eta _B, \eta _S)$ such that $\eta_B + (\eta_S)^+ S_0 - (\eta_S)^- (1-\lambda)S_0 > 0$, 
there exists a shadow price $\tilde{S}$, of the form $\tilde{S}_t = S_t e^{f(X_t)}$, where
  \begin{itemize}
	\item[-] $f(x)=\uy + \int_x^{\ox} \tfrac{g'(t)}{t}\, dt$, for
	  $x\in [\ux,\ox]$, 
	\item[-] $\hx= \left\{ \begin{array}{ll} \ox, & \textrm{if  }\ \frac{S_0 \eta_S}{\eta_B + S_0 \eta_S}>\overline{\pi}, \\ \ux, &\textrm{if  }\  \tfrac{S_0 \eta_S}{\eta_B + S_0 \eta_S}<\underline{\pi}, \\ \textrm{the solution to $\frac{(1-p)x}{p \, g(x)}=\frac{S_0 \eta_S \, e^{f(x)}}{\eta_B + S_0 \eta_S \, e^{f(x)}}$}, &\textrm{otherwise}, \end{array} \right.$
	\item[-]  $\{X_t\}_{t\in[0,\infty)}$ is the unique
		solution of reflected SDE 
\begin{displaymath}
\label{equ:Skorokhod}
\left\{   \begin{split}
     dX_t &= \Big( 
X_t \beta(\hat{\theta}(X_t))
+q \hat{\theta}(X_t) \hat{\Sigma}(X_t) \tfrac{X_t}{g'(X_t)}\Big) \, dt - \hat{\Sigma}(X_t) \tfrac{X_t}{g'(X_t)}\,   dB_t+ d\Phi_t,\\
X_0&=\hx.
   \end{split}
\right.
\end{displaymath}
where $\Phi$ is the instantaneous inward reflection term for the boundary $\{\ux,\ox\}$ and
\begin{displaymath}
\hat{\theta}(x)=-\tfrac{\sigma(1-p)x(qg'(x)-1)}{h(x)}, \,\, \hat{\Sigma}(x)=-\tfrac{\sigma(qg(x)-x)g'(x)}{h(x)}.
\end{displaymath}
  \end{itemize}
(3) The value $u(\eta_B,\eta_S)$ and an optimal trading strategy
 $(\hvpz,\hvp)$ of \eqref{primal} satisfy
$$u(\eta_B,\eta_S)= \tfrac{1}{p} \big(\eta_B + \eta_S S_0  e^{f(\hx)}\big)^p \vert g(\hx)\vert^{1-p}, \textrm{  and  } \tfrac{\hvp_t \tilde{S}_t}{\hvpz_t+\hvp_t \tilde{S}_t}=\tfrac{X_t}{q\, g(X_t)}. $$
\end{theorem}

\end{document}